\documentclass[11pt]{article}
\usepackage{fullpage}
\usepackage{amssymb}
\usepackage{amsmath}
\usepackage{graphicx}
\usepackage{enumerate}
\usepackage{hyperref}
\usepackage[hyperref]{ntheorem}
\usepackage[noadjust]{cite}

\usepackage{tikz}
\usetikzlibrary{decorations.pathreplacing}
\usetikzlibrary{decorations.pathmorphing}
\usetikzlibrary{calc}
\usetikzlibrary{shapes.geometric}
\usetikzlibrary{shapes.misc}
\usetikzlibrary{positioning}
\usetikzlibrary{patterns}
\usetikzlibrary{fit}
\usetikzlibrary{arrows}
\usetikzlibrary{decorations.markings}

\tikzstyle{every picture}+=[remember picture]

\usepackage{subfig}
\newtheorem{theorem}{Theorem}
\newtheorem{lemma}{Lemma}

\newtheorem{corollary}{Corollary}

\newcommand{\qed}{\hfill\ensuremath{\Box}\medskip\\\noindent}
\newenvironment{proof}{\noindent\emph{Proof. }}

\newcommand{\fingerprintq}{\ensuremath{\textsc{Fingerprint}}}
\newcommand{\lceq}{\ensuremath{\textsc{LCE}}}
\newcommand{\levelanc}{\ensuremath{\textsc{LA}}}
\renewcommand{\succ}{\ensuremath{\textsc{succ}}}
\newcommand{\pred}{\ensuremath{\textsc{pred}}}

\newcommand{\fp}{\ensuremath\phi}

\newcommand{\fpplus}{\ensuremath{\oplus}}
\newcommand{\fpdelsuffix}{\ensuremath{\ominus_s}}
\newcommand{\fpdelprefix}{\ensuremath{\ominus_p}}

\newcommand{\size}{\ensuremath{\mathit{size}}}
\newcommand{\depth}{\ensuremath{\mathit{depth}}}
\newcommand{\leaf}{\ensuremath{\mathit{leaf}}}

\newcommand{\heavypath}{\ensuremath{H}}
\newcommand{\rootnode}{\ensuremath{\mathit{root}}}
\newcommand{\lchild}{\ensuremath{\mathit{left}}}
\newcommand{\rchild}{\ensuremath{\mathit{right}}}
\newcommand{\leftnodes}{\ensuremath{V}} 
\newcommand{\leftsize}{\ensuremath{L}} 
\newcommand{\leftstring}{\ensuremath{P}} 
\newcommand{\str}{\ensuremath{S} }
\newcommand{\slp}{\ensuremath{G} }
\newcommand{\lslp}{\ensuremath{G_L} }

\newcommand{\lce}{\ensuremath\ell}

\title{Fingerprints in Compressed Strings\footnote{An extended abstract of this paper appeared at the 13th Algorithms and Data Structures Symposium.}}
\author{Philip Bille \\ \texttt{phbi@dtu.dk} \and Patrick Hagge Cording \\ \texttt{phaco@dtu.dk} \and Inge Li G{\o}rtz\thanks{Supported by a grant from the Danish Council for Independent Research $\vert$ Natural Sciences.} \\ \texttt{inge@dtu.dk} \and Benjamin Sach \\ \texttt{sach@dcs.warwick.ac.uk} \and Hjalte Wedel Vildh{\o}j \\ \texttt{hwvi@dtu.dk} \and S{\o}ren Vind\thanks{Supported by a grant from the Danish National Advanced Technology Foundation.} \\ \texttt{sovi@dtu.dk}}

\begin{document}
	
\maketitle

\begin{abstract}
\noindent The Karp-Rabin fingerprint of a string is a type of hash value that due to its strong properties has been used in many string algorithms. In this paper we show how to construct a data structure for a string $S$ of size $N$ compressed by a context-free grammar of size $n$ that answers fingerprint queries. That is, given indices $i$ and $j$, the answer to a query is the fingerprint of the substring $S[i,j]$. We present the first $O(n)$ space data structures that answer fingerprint queries without decompressing any characters. For Straight Line Programs (SLP) we get $O(\log N)$ query time, and for Linear SLPs (an SLP derivative that captures LZ78 compression and its variations) we get $O(\log \log N)$ query time. Hence, our data structures has the same time and space complexity as for random access in SLPs. We utilize the fingerprint data structures to solve the longest common extension problem in query time $O(\log N\log \lce)$ and $O(\log \lce \log\log \lce + \log\log N)$ for SLPs and Linear SLPs, respectively. Here, $\lce$ denotes the length of the LCE.
\end{abstract}

\section{Introduction}
Given a string \str of size $N$ and a Karp-Rabin fingerprint function $\fp$, the answer to a $\fingerprintq(i, j)$ query is the fingerprint $\fp(\str[i, j])$ of the substring $\str[i, j]$. We consider the problem of constructing a data structure that efficiently answers fingerprint queries when the string is compressed by a context-free grammar of size $n$.

The fingerprint of a string is an alternative representation that is much shorter than the string itself. By choosing the fingerprint function randomly at runtime it exhibits strong guarantees for the probability of two different strings having different fingerprints. Fingerprints were introduced by Karp and Rabin~\cite{karp1987efficient} and used to design a randomized string matching algorithm. Since then, they have been used as a central tool to design algorithms for a wide range of problems (see e.g.,~\cite{amir1992efficient, andoni2006efficient, cole2003faster, cormode2005substring, cormode2007string, farach1998string, gasieniec1996randomized, kalai2002efficient, porat2009exact}). 

A fingerprint requires constant space and it has the useful property that given the fingerprints $\fp(\str[1, i-1])$ and $\fp(\str[1, j])$, the fingerprint $\fp(\str[i, j])$ can be computed in constant time. By storing the fingerprints $\fp(\str[1, i])$ for $i=1\ldots N$ a query can be answered in $O(1)$ time. However, this data structure uses $O(N)$ space which can be exponential in $n$. Another approach is to use the data structure of G\c{a}sieniec~et~al.~\cite{gasieniec2005real} which supports linear time decompression of a prefix or suffix of the string generated by a node. To answer a query we find the deepest node that generates a string containing $S[i]$ and $S[j]$ and decompress the appropriate suffix of its left child and prefix of its right child. Consequently, the space usage is $O(n)$ and the query time is $O(h+j-i)$, where $h$ is the height of the grammar. The $O(h)$ time to find the correct node can be improved to $O(\log N)$ using the data structure by Bille et al.~\cite{bille2011random} giving $O(\log N + j-i)$ time for a $\fingerprintq(i, j)$ query. Note that the query time depends on the length of the decompressed string which can be large.

We present the first data structures that answers fingerprint queries on grammar compressed strings without decompressing any characters, and improve all of the above time-space trade-offs. Assume without loss of generality that the compressed string is given as a Straight Line Program (SLP). An SLP is a grammar in Chomsky normal form, i.e., each nonterminal has exactly two children. A Linear SLP is an SLP where the root is allowed to have more than two children, and for all other internal nodes, the right child must be a leaf. Linear SLPs capture the LZ78 compression scheme~\cite{lz78} and its variations. Our data structures give the following theorem.

\begin{theorem}\label{thm:fp}
	Let $S$ be a string of length $N$ compressed into an SLP $\slp$ of size~$n$. We can construct data structures that support $\fingerprintq$ queries in:
	\begin{enumerate}
		\item[(i)] $O(n)$ space and query time $O(\log N)$
		\item[(ii)] $O(n)$ space and query time $O(\log \log N)$ if $\slp$ is a Linear SLP
	\end{enumerate}
\end{theorem}

\noindent Hence, we show a data structure for fingerprint queries that has the same time and space complexity as for random access in SLPs. 

Our fingerprint data structures are based on the idea that a random access query for $i$ produces a path from the root to a leaf labelled $S[i]$. The concatenation of the substrings produced by the left children of the nodes on this path produce the prefix $S[1,i]$. We store the fingerprints of the strings produced by each node and concatenate these to get the fingerprint of the prefix instead. For \autoref{thm:fp}(i), we combine this with the fast random access data structure by Bille et al.~\cite{bille2011random}. For Linear SLPs we use the fact that the production rules form a tree to do large jumps in the SLP in constant time using a level ancestor data structure. Then a random access query is dominated by finding the node that produces $S[i]$ among the children of the root, which can be modelled as the predecessor problem.

Furthermore, we show how to obtain faster query time in Linear SLPs using finger searching techniques. Specifically, a finger for position $i$ in a Linear SLP is a pointer to the child of the root that produces $S[i]$.


\begin{theorem}\label{thm:ffp}
Let $S$ be a string of length $N$ compressed into an SLP $\slp$ of size~$n$. We can construct an $O(n)$ space data structure such that given a finger $f$ for position $i$ or $j$, we can answer a $\fingerprintq(i, j)$ query in time $O(\log \log D)$ where $D = |i - j|$.
\end{theorem}

\noindent Along the way we give a new and simple reduction for solving the finger predecessor problem on integers using any predecessor data structure as a black~box.

In compliance with all related work on grammar compressed strings, we assume that the model of computation is the RAM model with a word size of $\log N$ bits.


\subsection{Longest common extension in compressed strings}
As an application we show how to efficiently solve the longest common extension problem (LCE).
Given two indices $i, j$ in a string $S$, the answer to the $\lceq(i, j)$ query is the length $\lce$ of the maximum substring such that $S[i, i+\lce] = S[j, j+\lce]$. The compressed LCE problem is to preprocess a compressed string to support LCE queries. On uncompressed strings this is solvable in $O(N)$ preprocessing time, $O(N)$ space, and $O(1)$ query time with a nearest common ancestor data structure on the suffix tree for $S$ \cite{HT1984}. Other trade-offs are obtained by doing an exponential search over the fingerprints of strings starting in $i$ and $j$ \cite{bille12lce}. Using the exponential search in combination with the previously mentioned methods for obtaining fingerprints without decompressing the entire string we get $O((h+\lce)\log \lce)$ or $O((\log N+\lce)\log \lce)$ time using $O(n)$ space for an LCE query. Using our new (finger) fingerprint data structures and the exponential search we obtain \autoref{thm:lce}.

\begin{theorem}\label{thm:lce}
Let $G$ be an SLP of size $n$ that produces a string $S$ of length $N$. The SLP $G$ can be preprocessed in $O(N)$ time into a Monte Carlo data structure of size $O(n)$ that supports LCE queries on $S$ in
\begin{enumerate}
\item[(i)] $O(\log \ell \log N)$ time
\item[(ii)] $O(\log \lce \log\log \lce + \log\log N)$ time if $G$ is a Linear SLP.
\end{enumerate}
Here $\ell$ denotes the LCE value and queries are answered correctly with high probability. Moreover, a Las Vegas version of both data structures that always answers queries correctly can be obtained with
$O(N^2/n \log N)$ preprocessing time with high probability.
\end{theorem}



\noindent We furthermore show how to reduce the Las Vegas preprocessing time to $O(N\log N\log\log N)$ when all the internal nodes in the Linear SLP are children of the root (which is the case in LZ78).

The following corollary follows immediately because an LZ77 compression~\cite{lz77} consisting of $n$ phrases can be transformed to an SLP with $O(n\log\frac{N}{n})$ production rules~\cite{charikar2005smallest,rytter2003application}.

\begin{corollary}
	We can solve the $\lceq$ problem in $O(n \log \frac{N}{n})$ space and query time $O(\log \ell \log N)$ for LZ77 compression.
\end{corollary}

\noindent Finally, the LZ78 compression can be modelled by a Linear SLP $\lslp$ with constant overhead. Consider an LZ78 compression with $n$ phrases, denoted $r_1, \ldots, r_n$. A terminal phrase corresponds to a leaf in $\lslp$, and each phrase $r_j = (r_i, a)$, $i < j$, corresponds to a node $v \in \lslp$ with $r_i$ corresponding to the left child of $v$ and the right child of $v$ being the leaf corresponding to $a$. Therefore, we get the following corollary.

\begin{corollary}
	We can solve the $\lceq$ problem in $O(n)$ space and query time $O(\log \ell \log \log \ell + \log \log N)$ for LZ78 compression.
\end{corollary}

\section{Preliminaries}
Let $S = S[1, |S|]$ be a string of length $|S|$. Denote by $S[i]$ the character in $S$ at index $i$ and let $S[i, j]$ be the substring of $S$ of length $j - i+1$ from index $i \geq 1$ to $|S| \geq j \geq i$, both indices included.

A Straight Line Program (SLP) $\slp$ is a context-free grammar in Chomsky normal form that we represent as a node-labeled and ordered directed acyclic graph. Each leaf in $\slp$ is labelled with a character, and corresponds to a terminal grammar production rule. Each internal node in $\slp$ is labeled with a nonterminal rule from the grammar. The unique string $S(v)$ of length $\size(v) = |S(v)|$ is \emph{produced} by a depth-first left-to-right traversal of $v \in \slp$ and consist of the characters on the leafs in the order they are visited. We let $\rootnode(\slp)$ denote the root of $\slp$, and $\lchild(v)$ and $\rchild(v)$ denote the left and right child of an internal node $v \in \slp$, respectively.

A Linear SLP $\lslp$ is an SLP where we allow $\rootnode(\lslp)$ to have more than two children. All other internal nodes $v \in \lslp$ have a leaf as $\rchild(v)$. Although similar, this is not the same definition as given for the Relaxed SLP by Claude and Navarro~\cite{claude2011self}. The Linear SLP is more restricted since the right child of any node (except the root) must be a leaf. Any Linear SLP can be transformed into an SLP of at most double size by adding a new rule for each child of the root.

We extend the classic \emph{heavy path decomposition} of Harel and Tarjan \cite{HT1984} to SLPs as in \cite{bille2011random}. For each node $v \in \slp$, we select one edge from $v$ to a child with maximum size and call it the \emph{heavy edge}. The remaining edges are \emph{light edges}. Observe that $\size(u) \leq \size(v)/2$ if $v$ is a parent of $u$ and the edge connecting them is light. Thus, the number of light edges on any path from the root to a leaf is at most $O(\log N)$. 
A \emph{heavy path} is a path where all edges are heavy. The heavy path of a node $v$, denoted $\heavypath(v)$, is the unique path of heavy edges starting at $v$. Since all nodes only have a single outgoing heavy edge, the heavy path $\heavypath(v)$ and its leaf $\leaf(\heavypath(v))$, is well-defined for each node $v \in \slp$.


A \emph{predecessor data structure} supports predecessor and successor queries on a set $R \subseteq U = \{ 0, \ldots, N-1 \}$ of $n$ integers from a universe $U$ of size $N$. The answer to a \emph{predecessor query} $\pred(q)$ is the largest integer $r^- \in R$ such that $r^- \leq q$, while the answer to a \emph{successor query} $\succ(q)$ is the smallest integer $r^+ \in R$ such that $r^+ \geq q$. There exist predecessor data structures achieving a query time of $O(\log \log N)$ using space $O(n)$ \cite{van1976design, mehlhorn1990bounded, willard1983log}.

Given a rooted tree $T$ with $n$ vertices, we let $\depth(v)$ denote the length of the path from the root of $T$ to a node $v \in T$. A \emph{level ancestor data structure} on $T$ supports \emph{level ancestor queries} $\levelanc(v, i)$, asking for the ancestor $u$ of $v \in T$ such that $\depth(u) = \depth(v)-i$. There is a level ancestor data structure answering queries in $O(1)$ time using $O(n)$ space \cite{dietz1991finding} (see also \cite{berkman1994finding, alstrup2000improved, bender2004level}).

\subsection{Fingerprinting}
The Karp-Rabin fingerprint \cite{karp1987efficient} of a string $x$ is defined as $\fp(x) = \sum_{i=1}^{|x|} x[i] \cdot c^i \bmod p$, where $c$ is a randomly chosen positive integer, and $2N^{c+4} \leq p \leq 4N^{c+4}$ is a prime. Karp-Rabin fingerprints guarantee that given two strings $x$ and $y$, if $x = y$ then $\fp(x) = \fp(y)$. Furthermore, if $x \neq y$, then with high probability $\fp(x) \neq \fp(y)$. Fingerprints can be composed and subtracted as follows.

\begin{lemma}\label{lem:fp}
Let $x = y z$ be a string decomposable into a prefix $y$ and suffix $z$. Let $N$ be the maximum length of $x$, $c$ be a random integer and $2N^{c+4} \leq p \leq 4N^{c+4}$ be a prime. Given any two of the Karp-Rabin fingerprints $\fp(x)$, $\fp(y)$ and $\fp(z)$, it is possible to calculate the remaining fingerprint in constant time as follows:
\begin{align*}
	\fp(x) = \fp(y) \fpplus \fp(z) &= \fp(y) + c^{|y|} \cdot \fp(z) \bmod p \\
	\fp(y) = \fp(x) \fpdelsuffix \fp(z) &= \fp(x) - \frac{c^{|x|}}{c^{|z|}} \cdot \fp(z) \bmod p \\
	\fp(z) = \fp(x) \fpdelprefix \fp(y) &= \frac{\fp(x) - \fp(y)}{c^{|y|}} \bmod p
\end{align*}
\end{lemma}

\noindent In order to calculate the fingerprints of \autoref{lem:fp} in constant time, each fingerprint for a string $x$ must also store the associated exponent $c^{|x|} \bmod p$, and we will assume this is always the case. Observe that a fingerprint for any substring $\fp(S[i, j])$ of a string can be calculated by subtracting the two fingerprints for the prefixes $\fp(S[1, i-1])$ and $\fp(S[1, j])$. Hence, we will only show how to find fingerprints for prefixes in this paper.

\section{Basic fingerprint queries in SLPs}
We now describe a simple data structure for answering $\fingerprintq(1, i)$ queries for a string $S$ compressed into a SLP $\slp$ in time $O(h)$, where $h$ is the height of the parse tree for $S$. This method does not unpack the string to obtain the fingerprint, instead the fingerprint is generated by traversing $\slp$. 


The data structure stores $\size(v)$ and the fingerprint $\fp(S(v))$ of the string produced by each node $v \in \slp$. To compose the fingerprint $f = \fp(S[1, i])$ we start from the root of \slp and do the following. Let $v'$ denote the currently visited node, and let $p=0$ be a variable denoting the size the concatenation of strings produced by left children of visited nodes. We follow an edge to the right child of $v'$ if $p+\size(\lchild(v')) < i$, and follow a left edge otherwise. If following a right edge, update $f = f \fpplus \fp(S(\lchild(v')))$ such that the fingerprint of the full string generated by the left child of $v'$ is added to $f$, and set $p = p+\size(\lchild(v'))$. When following a left edge, $f$ and $p$ remains unchanged. When a leaf is reached, let $f = f \fpplus \fp(S(v'))$ to include the fingerprint of the terminal character. Aside from the concatenation of fingerprints for substrings, this procedure resembles a random access query for the character in position $i$ of $S$.

The procedure correctly composes $f = \fp(S[1, i])$ because the order in which the fingerprints for the substrings are added to $f$ is identical to the order in which the substrings are decompressed when decompressing $S[1, i]$. 

Since the fingerprint composition takes constant time per addition, the time spent generating a fingerprint using this method is bounded by the height of the parse tree for $S[i]$, denoted $O(h)$. Only constant additional space is spent for each node in $\slp$, so the space usage is $O(n)$.

\section{Faster fingerprints in SLPs}
Using the data structure of Bille et al.~\cite{bille2011random} to perform random access queries allows for a faster way to answer $\fingerprintq(1, i)$ queries.


\begin{lemma}[\cite{bille2011random}]\label{lem:slp:random}
Let $S$ be a string of length $N$ compressed into a SLP $\slp$ of size $n$. Given a node $v \in \slp$, we can support random access in $S(v)$ in $O(\log (\size(v)))$ time, at the same time reporting the sequence of heavy paths and their entry- and exit points in the corresponding depth-first traversal of $\slp(v)$.
\end{lemma}

\noindent The main idea is to compose the final fingerprint from substring fingerprints by performing a constant number of fingerprint additions per heavy path visited.

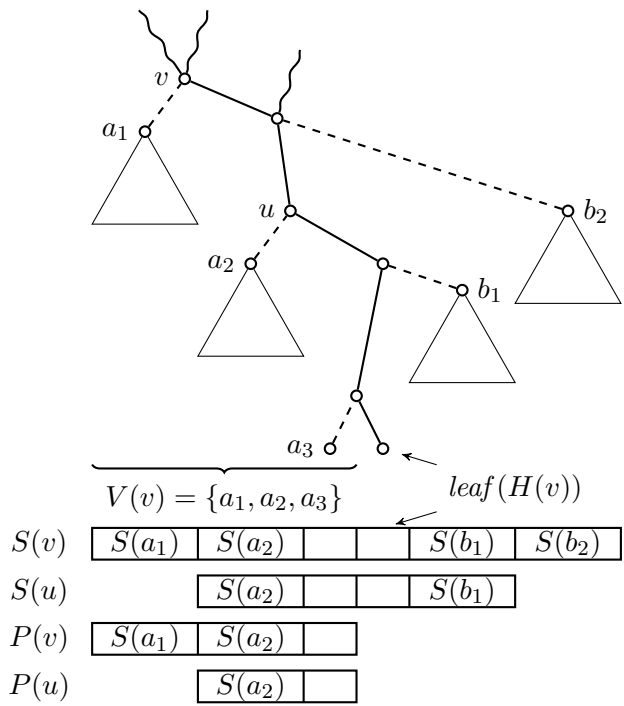
\begin{figure}[tb]
	
	\begin{center}
	\begin{tikzpicture}[x=1,y=1,-,>=stealth',auto, thick,
every label/.style={inner sep=3pt},
	tnode/.style={draw=black, fill=white, circle, inner sep=0pt, minimum size=4pt},
  	terminal/.style={circle,fill=white, inner sep=0.2pt},
  	box/.style={minimum height=11,minimum width=40, inner sep=1},
  	subtree/.style={thin,draw}]

\def\spacing{18}

\draw [] (60,-6) rectangle (120,6);
\foreach \x in {100}{
	\draw [] (\x,-6) -- (\x,6);
}

\draw [] (20,-6+1*\spacing) rectangle (120,6+1*\spacing);
\foreach \x in {60,100}{
	\draw [] (\x,-6+1*\spacing) -- (\x,6+1*\spacing);
}

\draw [] (20,-6+3*\spacing) rectangle (220,6+3*\spacing);
\foreach \x in {100,120,140}{
	\draw [] (\x,-6+2*\spacing) -- (\x,6+2*\spacing);
}

\draw [] (60,-6+2*\spacing) rectangle (180,6+2*\spacing);
\foreach \x in {60,100,120,140,180}{
	\draw [] (\x,-6+3*\spacing) -- (\x,6+3*\spacing);
}

\node[minimum width=40] at (0,2*\spacing) {$S(u)$};
\node[box] at (80,2*\spacing) {$S(a_2)$};
\node[box,minimum width=20] at (230,2*\spacing) {};
\node[box] at (160,2*\spacing) {$S(b_1)$};

\node[minimum width=40] at (0,3*\spacing) {$S(v)$};
\node[box] at (40,3*\spacing) {$S(a_1)$};
\node[box] at (80,3*\spacing) {$S(a_2)$};
\node[box,minimum width=20] at (230,3*\spacing) {};
\node[box] at (160,3*\spacing) {$S(b_1)$};
\node[box] at (200,3*\spacing) {$S(b_2)$};

\node[minimum width=40] at (0,\spacing) {$P(v)$};
\node[box] at (40,\spacing) {$S(a_1)$};
\node[box] at (80,\spacing) {$S(a_2)$};

\node[minimum width=40] at (0,0) {$P(u)$};
\node[box] at (80,0) {$S(a_2)$};


\def\nodeyoffset{10}

\node[tnode, label=left:$a_1$] (a) at (40, 200+\nodeyoffset) {};
\node[tnode, label=left:$a_2$] (b) at (80, 150+\nodeyoffset) {};
\node[tnode, label=left:$a_3$] (c) at (110, 80+\nodeyoffset) {};
\node[tnode, label=right:$b_2$] (d) at (200, 170+\nodeyoffset) {};
\node[tnode, label=right:$b_1$] (e) at (160, 140+\nodeyoffset) {};

\def\subtreeheight{35}

\path[subtree] (a) -- ($(a)+(-20,-\subtreeheight)$) -- ($(a)+(20,-\subtreeheight)$) -- (a);
\path[subtree] (b) -- ($(b)+(-20,-\subtreeheight)$) -- ($(b)+(20,-\subtreeheight)$) -- (b);
\path[subtree] (d) -- ($(d)+(-20,-\subtreeheight)$) -- ($(d)+(20,-\subtreeheight)$) -- (d);
\path[subtree] (e) -- ($(e)+(-20,-\subtreeheight)$) -- ($(e)+(20,-\subtreeheight)$) -- (e);


\node[tnode, label=left:$v$] (f) at ($(a)+(15, 20)$) {};
\node[tnode] (g) at ($(d)+(-110, 35)$) {};
\node[tnode, label=left:$u$] (h) at ($(b)+(15, 20)$) {};
\node[tnode] (i) at ($(e)+(-30, 10)$) {};
\node[tnode] (j) at ($(c)+(10, 20)$) {};
\node[tnode] (k) at ($(j)+(10, -20)$) {};

\path[draw] (f) -- (g) -- (h) -- (i) -- (j) -- (k);
\path[draw, dashed] (a) -- (f);
\path[draw, dashed] (b) -- (h);
\path[draw, dashed] (c) -- (j);
\path[draw, dashed] (d) -- (g);
\path[draw, dashed] (e) -- (i);

\draw[decorate,decoration={brace,raise=0pt,amplitude=3pt}] (120,84) -- (20,84) node[midway,label={[label distance=-4]below:$V(v)=\{a_1,a_2,a_3\}$}] {};

\node[] (leaf) at ($(k)+(50,-15)$) {$\leaf(H(v))$};
\path[->,>=stealth',auto,thin,shorten >=5] (leaf) edge [] (k);
\path[->,>=stealth',auto,thin,shorten >=5] (leaf) edge [] (130,6+3*\spacing);

\node[] (d1) at ($(f)+(-20,30)$) {};
\node[] (d2) at ($(f)+(10,30)$) {};
\node[] (d3) at ($(g)+(10,30)$) {};

\path[draw,decorate,decoration={snake,amplitude=0.3mm,segment length=4mm,pre length=0mm, post length=0mm}] (d1) -- (f);
\path[draw,decorate,decoration={snake,amplitude=0.3mm,segment length=4mm,pre length=0mm, post length=0mm}] (d2) -- (f);
\path[draw,decorate,decoration={snake,amplitude=0.3mm,segment length=4mm,pre length=0mm, post length=0mm}] (d3) -- (g);

\end{tikzpicture}	
	\end{center}
	\caption{Figure showing how $S(v)$ and its prefix $\leftstring(v)$ is composed of substrings generated by the left children $a_1, a_2, a_3$ and right children $b_1, b_2$ of the heavy path $\heavypath(v)$. Also illustrates how this relates to $S(u)$ and $\leftstring(u)$ for a node $u \in \heavypath(v)$.\label{fig:slp:strings}}
\end{figure}

In order to describe the data structure, we will use the following notation. Let $\leftnodes(v)$ be the left children of the nodes in $\heavypath(v)$ where the heavy path was extended to the right child, ordered by increasing depth. The order of nodes in $V(v)$ is equal to the sequence in which they occur when decompressing $S(v)$, so the concatenation of the strings produced by nodes in $V(v)$ yields the prefix $P(v)=S(v)[1,L(v)]$, where $\leftsize(v) = \sum_{u \in \leftnodes(v)} \size(u)$. Observe that $\leftstring(u)$ is a suffix of $\leftstring(v)$ if $u \in \heavypath(v)$. See \autoref{fig:slp:strings} for the relationship between $u$, $v$ and the defined strings.
 
Let each node $v \in \slp$ store its unique outgoing heavy path $\heavypath(v)$, the length $\leftsize(v)$, $\size(v)$, and the fingerprints $\fp(\leftstring(v))$ and $\fp(S(v))$. By forming heavy path trees of total size $O(n)$ as in \cite{bille2011random}, we can store $\heavypath(v)$ as a pointer to a node in a heavy path tree (instead of each node storing the full sequence).

The fingerprint $f = \fp(S[1, i])$ is composed from the sequence of heavy paths visited when performing a single random access query for $S[i]$ using \autoref{lem:slp:random}. Instead of adding all left-children of the path towards $S[i]$ to $f$ individually, we show how to add all left-children hanging from each visited heavy path in constant time per heavy path. Thus, the time taken to compose $f$ is $O(\log N)$. 

More precisely, for the pair of entry- and exit-nodes $v, u$ on each heavy path $H$ traversed from the root to $S[i]$, we set $f = f \fpplus (\fp(\leftstring(v)) \fpdelsuffix \fp(\leftstring(u))$ (which is allowed because $\leftstring(u)$ is a suffix of $\leftstring(v)$). If we leave $u$ by following a right-pointer, we additionally set $f = f \fpplus \fp(S(\lchild(u)))$. If $u$ is a leaf, set $f = f \fpplus \fp(S(u))$ to include the fingerprint of the terminal character. 

Remember that $\leftstring(v)$ is exactly the string generated from $v$ along $H$, produced by the left children of nodes on $H$ where the heavy path was extended to the right child. Thus, this method corresponds exactly to adding the fingerprint for the substrings generated by all left children of nodes on $H$ between the entry- and exit-nodes in depth-first order, and the argument for correctness from the slower fingerprint generation also applies here.

Since the fingerprint composition takes constant time per addition, the time spent generating a fingerprint using this method is bounded by the number of heavy paths traversed, which is $O(\log N)$. Only constant additional space is spent for each node in $\slp$, so the space usage is $O(n)$. This concludes the proof of \autoref{thm:fp}(i).

\section{Faster fingerprints in Linear SLPs}
In this section we show how to quickly answer $\fingerprintq(1, i)$ queries on a Linear SLP $\lslp$. In the following we denote the sequence of $k$ children of $\rootnode(\lslp)$ from left to right by $r_1, \ldots, r_k$. Also, let $R(j) = \sum_{m=1}^j \size(r_m)$ for $j = 0, \ldots, k$. That is, $R(j)$ is the length of the prefix of $S$ produced by $\lslp$ including $r_j$ (and $R(0)$ is the empty prefix).

We also define the dictionary tree $F$ over $\lslp$ as follows. Each node $v \in \lslp$ corresponds to a single vertex $v^F \in F$. There is an edge $(u^F, v^F)$ labeled $c$ if $u = \lchild(v)$ and $c = S(\rchild(v))$. If $v$ is a leaf, there is an edge $(\rootnode(F), v^F)$ labeled $S(v)$. That is, a left child edge of $v \in \lslp$ is converted to a parent edge of $v^F \in F$ labeled like the right child leaf of $v$. Note that for any node $v \in \lslp$ except the root, producing $S(v)$ is equivalent to following edges and reporting edge labels on the path from $\rootnode(F)$ to $v^F$. Thus, the prefix of length $a$ of $S(v)$ may be produced by reporting the edge labels on the path from $\rootnode(F)$ until reaching the ancestor of $v^F$ at depth $a$.

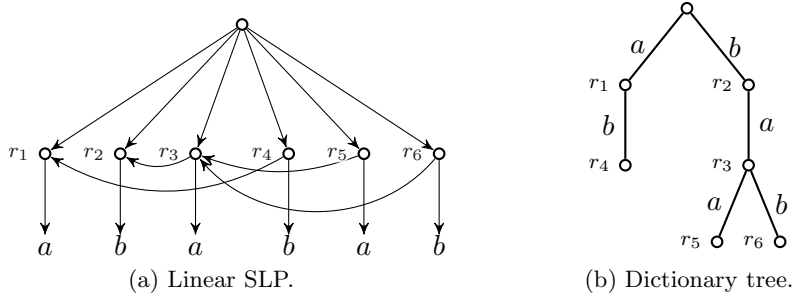
\begin{figure}[tb]
\begin{center}
\subfloat[Linear SLP.]{
\pgfdeclarelayer{background layer}
\pgfdeclarelayer{foreground layer}
\pgfsetlayers{background layer,main,foreground layer}
\begin{tikzpicture}[->,>=stealth',auto, thick,
every label/.style={rectangle, rounded corners, font=\scriptsize, line width=5pt, inner sep=1pt, fill=white, fill opacity=0.9, text opacity=1.0},
	main node/.style={draw=black, fill=white, circle, inner sep=0pt, minimum size=4pt},
  	terminal/.style={rectangle,text height=7pt, fill=white, inner sep=0.2pt}]

  
\node[main node] (root) [] {};
  \node[main node, label=left:$r_1$] (1) [below left=1.6cm and 2.5cm of root] {};
  \node[main node, label=left:$r_2$] (2) [below left=1.6cm and 1.5cm of root] {};
  \node[main node, label=left:$r_3$] (3) [below left=1.6cm and 0.5cm of root] {};
  \node[main node, label=left:$r_4$] (4) [below right=1.6cm and 0.5cm of root] {};
  \node[main node, label=left:$r_5$] (5) [below right=1.6cm and 1.5cm of root] {};
  \node[main node, label=left:$r_6$] (6) [below right=1.6cm and 2.5cm of root] {};

  \node[terminal] (a1) [below=1cm of 1] {$a$};
  \node[terminal] (a2) [below=1cm of 3] {$a$};
  \node[terminal] (a3) [below=1cm of 5] {$a$};
  \node[terminal] (b1) [below=1cm of 2] {$b$};
  \node[terminal] (b2) [below=1cm of 4] {$b$};
  \node[terminal] (b3) [below=1cm of 6] {$b$};
  
  \begin{pgfonlayer}{background layer}
  \path[] (root) edge [] (1);
  \path[] (root) edge [] (2);
  \path[] (root) edge [] (3);
  \path[] (root) edge [] (4);
  \path[] (root) edge [] (5);
  \path[] (root) edge [] (6);

  \path[] (1) edge [] (a1);
  \path[] (3) edge [] (a2);
  \path[] (5) edge [] (a3);
  \path[] (2) edge [] (b1);
  \path[] (4) edge [] (b2);
  \path[] (6) edge [] (b3);

  \path[] (3) edge [bend left] (2);
  \path[] (4) edge [bend left] (1);
  \path[] (5) edge [bend left=20] (3);
  \path[] (6) edge [bend left=50] (3);

\end{pgfonlayer}{background layer}
\end{tikzpicture}
}\quad\quad\quad\quad\subfloat[Dictionary tree.]{
\begin{tikzpicture}[-,>=stealth',auto, thick,
every label/.style={rectangle, font=\scriptsize, inner sep=3pt},
	main node/.style={draw=black, fill=white, circle, inner sep=0pt, minimum size=4pt},
  	terminal/.style={circle,fill=white, inner sep=0.2pt}]

  \node[main node] (root) [] {};
  \node[main node, label=left:$r_1$] (1) [below left=0.9cm and 0.7cm of root] {};
  \node[main node, label=left:$r_2$] (2) [below right=0.9cm and 0.7cm of root] {};
  \node[main node, label=left:$r_3$] (3) [below=0.9cm of 2] {};
  \node[main node, label=left:$r_4$] (4) [below=0.9cm of 1] {};
  \node[main node, label=left:$r_5$] (5) [below left=0.9cm and 0.3cm of 3] {};
  \node[main node, label=left:$r_6$] (6) [below right=0.9cm and 0.3cm of 3] {};

  \path[] (root) edge [] node [left] {$a$} (1);
  \path[] (root) edge [] node [right] {$b$} (2);
  \path[] (1) edge [] node [left] {$b$} (4);
  \path[] (3) edge [] node [left] {$a$} (5);
  \path[] (2) edge [] node [right] {$a$} (3);
  \path[] (3) edge [] node [right] {$b$} (6);
\end{tikzpicture}	
}

\caption{A Linear SLP compressing the string \texttt{abbaabbaabab} and the dictionary tree obtained from the Linear SLP.}\label{fig:linslpex1}

\end{center}
\end{figure}

The data structure stores a predecessor data structure over the prefix lengths $R(j)$ and the associated node $r_j$ and fingerprint $\fp(S[1, R(j)])$ for $j = 0, \ldots, k$. We also have a doubly linked list of all $r_j$'s with bidirectional pointers to the predecessor data structure and $\lslp$.  We store the dictionary tree $F$ over $\lslp$, augment it with a level ancestor data structure, and add bidirectional pointers between $v \in \lslp$ and $v^F \in F$. Finally, for each node $v \in \lslp$, we store the fingerprint of the string it produces, $\fp(S(v))$. 

A query $\fingerprintq(1, i)$ is answered as follows. Let $R(m)$ be the predecessor of $i$ among $R(0), R(1), \ldots, R(k)$. Compose the answer to $\fingerprintq(1, i)$ from the two fingerprints $\fp(S[1, R(m)])\fpplus\fp(S[R(m)+1, i])$. The first fingerprint $\fp(S[1, R(m)])$ is stored in the data structure and the second fingerprint $\fp(S[R(m)+1, i])$ can be found as follows. Observe that $S[R(m)+1, i]$ is fully generated by $r_{m+1}$ and hence a prefix of $S(r_{m+1})$ of length $i-R(m)$. We can get $r_{m+1}$ in constant time from $r_m$ using the doubly linked list. We use a level ancestor query $u^F = \levelanc(r_{m+1}^F, i-R(m))$ to determine the ancestor of $r_{m+1}^F$ at depth $i-R(m)$, corresponding to a prefix of $r_{m+1}$ of the correct length. From $u_F$ we can find $\fp(S(u)) = \fp(S[R(m)+1, i])$.

It takes constant time to find $\fp(S[R(m)+1, i])$ using a single level ancestor query and following pointers. Thus, the time to answer a query is bounded by the time spent determining $\fp(S[1, R(m)])$, which requires a predecessor query among $k$ elements (i.e. the number of children of $\rootnode(\lslp)$) from a universe of size $N$. The data structure uses $O(n)$ space, as there is a bijection between nodes in $\lslp$ and vertices in $F$, and we only spend constant additional space per node in $\lslp$ and vertex in $F$. This concludes the proof of \autoref{thm:fp}(ii).

\section{Finger fingerprints in Linear SLPs}
The $O(\log\log N)$ running time of a $\fingerprintq(1, i)$ query is dominated by having to find the predecessor $R(m)$ of $i$ among $R(0), R(1), \ldots, R(k)$. Given $R(m)$ the rest of the query takes constant time. In the following, we show how to improve the running time of a $\fingerprintq(1, i)$ query to $O(\log \log |j-i|)$ given a finger for position $j$.  
Recall that a finger $f$ for  a position $j$ is a pointer to the node $r_m$ producing $S[j]$.
To achieve this, we present a simple linear space finger predecessor data structure that is interchangeable with any other predecessor data structure.

\subsection{Finger Predecessor}
Let $R \subseteq U = \{ 0, \ldots, N-1 \}$ be a set of $n$ integers from a universe $U$ of size $N$. Given a finger $f \in R$ and a query point $q \in U$, the \emph{finger predecessor problem} is to answer finger predecessor or successor queries in time depending on the universe distance $D = |f-q|$ from the finger to the query point. 
Belazzougui et al.~\cite{fingerpred} present a succinct solution for solving the finger predecessor problem relying on a modification of z-fast tries. Here, we use a simple reduction for solving the finger predecessor problem using any predecessor data structure as a black box. 

\begin{lemma}\label{lem:fingerpred}
	Let $R \subseteq U = \{ 0, \ldots, N-1 \}$ be a set of $n$ integers from a universe $U$ of size $N$. 
	Given a predecessor data structure with query time $t(N, n)$ using $s(N, n)$ space,
	we can solve the finger predecessor problem in time $O(t(D, n))$ using space $O(s(N, \frac{n}{\log N}) \log N)$.
\end{lemma}
\begin{proof}
Construct a complete balanced binary search tree $T$ over the universe $U$. The leaves of $T$ represent the integers in $U$, and we say that a vertex \emph{span} the range of $U$ represented by the leaves in its subtree. Mark the leaves of $T$ representing the integers in $R$. We remove all vertices in $T$ where the subtree contains no marked vertices. Observe that a vertex at height $j$ span a universe range of size $O(2^j)$. We augment $T$ with a level ancestor data structure answering queries in constant time. Finally, left- and right-neighbour pointers are added for all nodes in $T$.

Each internal node $v \in T$ at height $j$ store an instance of the given predecessor data structure for the set of marked leaves in the subtree of $v$. The size of the universe for the predecessor data structure equals the span of the vertex and is $O(2^j)$\footnote{The integers stored by the data structure may be shifted by some constant $k \cdot 2^j$ for a vertex at height $j$, but we can shift all queries by the same constant and thus the size of the universe is $2^j$.}.

Given a finger $f \in R$ and a query point $q \in U$, we will now describe how to find both $\succ(q)$ and $\pred(q)$ when $q < f$. The case $q > f$ is symmetric. 
Observe that $f$ corresponds to a leaf in $T$, denoted $f_l$. We answer a query by determining the ancestor $v$ of $f_l$ at height $h = \lceil \log(|f - q|) \rceil$ and its left neighbour $v_L$ (if it exists). We query for $\succ(q)$ in the predecessor data structures of both $v$ and $v_L$, finding at least one leaf in $T$ (since $v$ spans $f$ and $q < f$). We return the leaf representing the smallest result as $\succ(q)$ and its left neighbour in $T$ as $\pred(q)$.

Observe that the predecessor data structures in $v$ and $v_L$ each span a universe of size $O(2^h) = O(|f-q|) = O(D)$. All other operations performed take constant time. Thus, for a predecessor data structure with query time $t(N,n)$, we can answer finger predecessor queries in time $O(t(D, n))$.

The height of $T$ is $O(\log N)$, and there are $O(n \log N)$ vertices in $T$ (since vertices spanning no elements from $R$ are removed). Each element from $R$ is stored in $O(\log N)$ predecessor data structures. Hence, given a predecessor data structure with space usage $s(N, n)$, the total space usage of the data structure is $O(s(N, n) \log N)$. 

We reduce the size of the data structure by reducing the number of elements it stores to $O(\frac{n}{\log N})$. This is done by partitioning $R$ into $O(\frac{n}{\log N})$ sets of consecutive elements $R_i$ of size $O(\log N)$. We choose the largest integer in each $R_i$ set as the representative $g_i$ for that set, and store that in the data structure described above. We store the integers in set $R_i$ in an atomic heap \cite{fredmanwillardfusion, Hagerup1998} capable of answering predecessor queries in $O(1)$ time and linear space for a set of size $O(\log N)$. Each element in $R$ keep a pointer to the set $R_i$ it belongs to, and each set left- and right-neighbour pointers.

Given a finger $f \in R$ and a query point $q \in U$, we describe how to determine $\pred(q)$ and $\succ(q)$ when $q < f$. The case $q > f$ is symmetric. We first determine the closest representatives $g_l$ and $g_r$ on the left and right of $f$, respectively. Assuming $q < g_l$, we proceed as before using $g_l$ as the finger into $T$ and query point $q$. This gives $p = \pred(q)$ and $s = \succ(q)$ among the representatives. If $g_l$ is undefined or $g_l < q < f \leq g_r$, we select $p = g_l$ and $s = g_r$.
To produce the final answers, we perform at most 4 queries in the atomic heaps that $p$ and $s$ are representatives for. 

All queries in the atomic heaps take constant time, and we can find $g_l$ and $g_r$ in constant time by following pointers. If we query a predecessor data structure, we know that the range it spans is $O(|g_l - q|) = O(|f-q|) = O(D)$ since $q < g_l < f$. Thus, given a predecessor data structure with query time $t(N, n)$, we can solve the finger predecessor problem in time $O(t(D, n))$.

The total space spent on the atomic heaps is $O(n)$ since they partition $R$. The number of representatives is $O(\frac{n}{\log N})$. Thus, given a predecessor data structure with space usage $s(N, n)$, we can solve the finger predecessor problem in space $O(s(N, \frac{n}{\log N}) \log N)$.\qed
\end{proof}

\noindent Using the 
van Emde Boas predecessor data structure \cite{van1976design, mehlhorn1990bounded, willard1983log} with $t(N, n) = O(\log \log N)$ query time using $s(N, n) = O(n)$ space, we 
obtain the following corollary.

\begin{corollary}\label{cor:fingerpred}
	Let $R \subseteq U = \{ 0, \ldots, N-1 \}$ be a set of $n$ integers from a universe $U$ of size $N$. 
	Given a finger $f \in R$ and a query point $q \in U$, we can solve the finger predecessor problem in time $O(\log \log |f - q|)$ and space $O(n)$.
\end{corollary}

\subsection{Finger Fingerprints}
We can now prove Theorem~\ref{thm:ffp}. Assume wlog that we have a finger for $i$, i.e., we  are given a finger $f$ to the node $r_m$ generating $S[i]$. From this we can in constant time get a pointer to $r_{m+1}$ in the doubly linked list and from this a pointer to $R(m+1)$ in the predecessor data structure. If $R(m+1) > j$ then $R(m)$ is the predecessor of $j$. Otherwise,  using Corollary~\ref{cor:fingerpred} we can in time $O(\log\log |R(m+1)-j|)$ find the predecessor of $j$. Since $R(m+1) \geq i$ and the rest of the query takes constant time, the total time for the query is $O(\log \log |i-j|)$.

\section{Longest Common Extensions in Compressed Strings}
Given an SLP $\slp$, the longest common extension (LCE) problem is to build a data structure for $\slp$ that supports longest common extension queries $LCE(i,j)$. In this section we show how to use our fingerprint data structures as a tool for doing LCE queries and hereby obtain \autoref{thm:lce}.


\subsection{Computing Longest Common Extensions with Fingerprints}

We start by showing the following general lemma that establishes the connection between LCE and fingerprint queries.

\begin{lemma}\label{lem:lce-comparisons}
For any string $S$ and any partition $S=s_1 s_2 \cdots s_t$ of $S$ into $k$ non-empty substrings called phrases, $\ell=\lceq(i,j)$ can be found by comparing $O(\log \ell)$ pairs of substrings of $S$ for equality. Furthermore, all substring comparisons $x=y$ are of one of the following two types:
\begin{description}
\item[Type 1] Both $x$ and $y$ are fully contained in (possibly different) phrase substrings.
\item[Type 2] $|x|=|y|=2^p$ for some $p=0,\ldots,\log(\ell)+1$ and for $x$ or $y$ it holds that
\begin{enumerate}
\item[(a)] The start position is also the start position of a phrase substring, or
\item[(b)] The end position is also the end position of a phrase substring.
\end{enumerate}
\end{description}
\end{lemma}
\begin{proof}
Let a position of $S$ be a \emph{start} (\emph{end}) position if a phrase starts (ends) at that position. Moreover, let a comparison of two substrings be of \emph{type 1} (\emph{type 2}) if it satisfies the first (second) property in the lemma. We now describe how to find $\ell = \lceq(i,j)$ by using $O(\log \ell)$ type 1 or 2 comparisons.

If $i$ or $j$ is not a start position, we first check if $S[i,i+k] = S[j,j+k]$ (type 1), where $k \geq 0$ is the minimum integer such that $i+k$ or $j+k$ is an end position. If the comparison fails, we have restricted the search for $\ell$ to two phrase substrings, and we can find the exact value using $O(\log \ell)$ type 1 comparisons.

Otherwise, $\lceq(i,j) = k + \lceq(i+k+1,j+k+1)$ and either $i+k+1$ or $j+k+1$ is a start position. This leaves us with the task of describing how to answer $\lceq(i,j)$, assuming that either $i$ or $j$ is a start position.

We first use $p=O(\log \ell)$ type 2 comparisons to determine the biggest integer $p$ such that $S[i , i+2^p] = S[j , j+2^p]$. It follows that $\ell \in [2^p, 2^{p+1}]$. Now let $q<2^p$ denote the length of the longest common prefix of the substrings $x=S[i+2^p+1 , i+2^{p+1}]$ and $y=S[j+2^p+1 , j+2^{p+1}]$, both of length $2^p$. Clearly, $\ell = 2^p + q$. By comparing the first half $x'$ of $x$ to the first half $y'$ of $y$, we can determine if $q \in [0,2^{p-1}]$ or $q \in [2^{p-1}+1,2^p-1]$. By recursing we obtain the exact value of $q$ after $\log 2^p = O(\log \ell)$ comparisons.

However, comparing $x'=S[a_1,b_1]$ and $y'=S[a_2,b_2]$ directly is not guaranteed to be of type 1 or 2. To fix this, we compare them indirectly using a type 1 and type 2 comparison as follows. Let $k < 2^p$ be the minimum integer such that $b_1-k$ or $b_2-k$ is a start position. If there is no such $k$ then we can compare $x'$ and $y'$ directly as a type 1 comparison. Otherwise, it holds that $x' = y'$ if and only if $S[b_1-k,b_1] = S[b_2-k,b_2]$ (type 1) and $S[a_1-k-1,b_1-k-1] = S[a_2-k-1,b_2-k-1]$ (type 2).
\qed
\end{proof}

\noindent \autoref{thm:lce} follows by using fingerprints to perform the substring comparisons. In particular, we obtain a Monte Carlo data structure that can answer a LCE query in $O(\log\ell \log N)$ time for SLPs and in $O(\log \ell \log\log N)$ time for Linear SLPs. In the latter case, we can use \autoref{thm:ffp} to reduce the query time to $O(\log\ell\log\log\ell + \log\log N)$ by observing that for all but the first fingerprint query, we have a finger into the data structure.


\subsection{Verifying the Fingerprint Function}
Since the data structure is Monte Carlo, there may be collisions among the fingerprints used to determine the LCE, and consequently the answer to a query may be incorrect. We now describe how to obtain a Las Vegas data structure that always answers LCE queries correctly. We do so by showing how to efficiently verify that the fingerprint function $\fp$ is $\emph{good}$, i.e., collision-free on all substrings compared in the computation of $\lceq(i,j)$. We give two verification algorithms. One that works for LCE queries in SLPs, and a faster one that works for Linear SLPs where all internal nodes are children of the root (e.g. LZ78).

%

\subsubsection{SLPs}
If we let the phrases of $S$ be its individual characters, we can assume that all fingerprint comparisons are of type 2 (see \autoref{lem:lce-comparisons}). We thus only have to check that $\fp$ is collision-free among all substrings of length $2^p,p=0,\ldots,\log N$. We verify this in $\log N$ rounds. In round $p$ we maintain the fingerprint of a sliding window of length $2^p$ over $S$. For each substring $x$ we insert $\fp(x)$ into a dictionary. If the dictionary already contains a fingerprint $\fp(y)=\fp(x)$, we verify that $x=y$ in constant time by checking if $\fp(x[1,2^{p-1}]) = \fp(y[1,2^{p-1}])$ and $\fp(x[2^{p-1}+1,2^p]) = \fp(y[2^{p-1}+1,2^{p}])$. This works because we have already verified that the fingerprinting function is collision-free for substrings of length $2^{p-1}$. Note that we can assume that for any fingerprint $\fp(x)$ the fingerprints of the first and last half of $x$ are available in constant time, since we can store and maintain these at no extra cost. In the first round $p=0$, we check that $x=y$ by comparing the two characters explicitly. If $x \neq y$ we have found a collision and we abort and report that $\fp$ is not good. If all rounds are successfully verified, we report that $\fp$ is good.

For the analysis, observe that computing all fingerprints of length $2^p$ in the sliding window can be implemented by a single traversal of the SLP parse tree in $O(N)$ time. Thus, the algorithm correctly decides whether $\fp$ is good in $O(N\log N)$ time and $O(N)$ space. We can easily reduce the space to $O(n)$ by carrying out each round in $O(N/n)$ iterations, where no more than $n$ fingerprints are stored in the dictionary in each iteration. So, alternatively, $\fp$ can be verified in $O(N^2/n \log N)$ time and $O(n)$ space.

\subsubsection{Linear SLPs} In Linear SLPs where all internal nodes are children of the root, we can reduce the verification time to $O(N\log N\log\log N)$, while still using $O(n)$ space. To do so, we use \autoref{lem:lce-comparisons} with the partition of $S$ being the root substrings. We verify that $\fp$ is collision-free for type 1 and type 2 comparisons separately.

\paragraph{Type 1 Comparisons.}
We carry out the verification in rounds. In round $p$ we check that no collisions occur among the $p$-length substrings of the root substrings as follows: We traverse the SLP maintaining the fingerprint of all $p$-length substrings. For each substring $x$ of length $p$, we insert $\fp(x)$ into a dictionary. If the dictionary already contains a fingerprint $\fp(y) = \fp(x)$ we verify that $x=y$ in constant time by checking if $x[1] = y[1]$ and $\fp(x[2,|x|]) = \fp(y[2,|y|])$ (type 1).

Every substring of a root substring ends in a leaf in the SLP and is thus a suffix of a root substring. Consequently, they can be generated by a bottom up traversal of the SLP. The substrings of length 1 are exactly the leaves. Having generated the substrings of length $p$, the substrings of length $p+1$ are obtained by following the parents left child to another root node and prepending its right child. In each round the $p$ length substrings correspond to a subset of the root nodes, so the dictionary never holds more than $n$ fingerprints. Furthermore, since each substring is a suffix of a root substring, and the root substrings have at most $N$ suffixes in total, the algorithm will terminate in $O(N)$ time.

\paragraph{Type 2 Comparisons.}
We adopt an approach similar to that for SLPs and verify $\fp$ in $O(\log N)$ rounds. In round $p$ we store the fingerprints of the substrings of length $2^p$ that start or end at a phrase boundary in a dictionary. We then slide a window of length $2^p$ over $S$ to find the substrings whose fingerprint equals one of those in the dictionary. Suppose the dictionary in round $p$ contains the fingerprint $\fp(y)$, and we detect a substring $x$ such that $\fp(x)=\fp(y)$. To verify that $x=y$, assume that $y$ starts at a phrase boundary (the case when it ends in a phrase boundary is symmetric). As before, we first check that the first half of $x$ is equal to the first half of $y$ using fingerprints of length $2^{p-1}$, which we know are collision-free. Let $x'=S[a_1,b_1]$ and $y'=S[a_2,b_2]$ be the second half of $x$ and $y$. Contrary to before, we can not directly compare $\fp(x')=\fp(y')$, since neither $x'$ nor $y'$ is guaranteed to start or end at a phrase boundary. Instead, we compare them indirectly using a type 1 and type 2 comparison as follows: Let $k < 2^{p-1}$ be the minimum integer such that $b_1-k$ or $b_2-k$ is a start position. If there is no such $k$ then we can compare $x'$ and $y'$ directly as a type 1 comparison. Otherwise, it holds that $x' = y'$ if and only if $\fp(S[b_1-k,b_1]) = \fp(S[b_2-k,b_2])$ (type 1) and $\fp(S[a_1-k-1,b_1-k-1]) = \fp(S[a_2-k-1,b_2-k-1])$ (type 2), since we have already verified that $\fp$ is collision-free for type 1 comparisions and type 2 comparisions of length $2^{p-1}$.

The analysis is similar to that for SLPs. The sliding window can be implemented in $O(N)$ time, but for each window position we now need $O(\log\log N)$ time to retrieve the fingerprints, so the total time to verify $\fp$ for type 2 collisions becomes $O(N \log N \log\log N)$. The space is $O(n)$ since in each round the dictionary stores at most $O(n)$ fingerprints.

\bibliographystyle{abbrv}

\bibliography{references}

\end{document}